\documentclass[%
 amsmath,amssymb,
 aps,
]{revtex4-2}
\usepackage[ruled,vlined]{algorithm2e}
\usepackage{graphicx}
\usepackage{dcolumn}
\usepackage{bm}
\usepackage{amsthm} 
\usepackage[normalem]{ulem} 

\usepackage{caption}
\usepackage{subcaption}
\usepackage{hyperref}
\usepackage{xcolor}

\DeclareMathOperator*{\argmax}{arg\,max}

\begin{document}

\preprint{APS/123-QED}

\title{How fragile is your network? More than you think.}

\author{Jeremie Fish}
\email{jafish@clarkson.edu}
 \affiliation{%
 Department of Electrical and Computer Engineering, Clarkson University
}%
\affiliation{Clarkson Center for Complex Systems Science}
\author{Mahesh Banavar}%
\affiliation{%
 Department of Electrical and Computer Engineering, Clarkson University
}%
\affiliation{Clarkson Center for Complex Systems Science}

\author{Erik Bollt}
\affiliation{%
 Department of Electrical and Computer Engineering, Clarkson University
}%
\affiliation{Clarkson Center for Complex Systems Science}
\theoremstyle{theorem}
\newtheorem{theorem}{Theorem}[section]

\newtheorem{corollary}{Corollary}[theorem]

\theoremstyle{definition}
\newtheorem{definition}{Definition}[section]

\theoremstyle{lemma}
\newtheorem{lemma}{Lemma}[section]

\date{\today}

\begin{abstract}

Graphs are pervasive in our everyday lives, with relevance to biology, the internet, and infrastructure, as well as numerous other applications. It is thus necessary to have an understanding of how quickly a graph disintegrates, whether by random failure or by targeted attack. While much of the interest in this subject has been focused on targeted removal of nodes, there has been some recent interest in targeted edge removal. Here, we focus on how robust a graph is against edge removal. We define a measure of network fragility, {\color{black} which allows us to classify networks in terms of fragility, defined as} the fraction of edges removed to the largest connected component.  We construct a class of graphs that is robust to edge removal {\color{black}under our classification system}. Furthermore, it is demonstrated that graphs generally disintegrate faster than would be anticipated by a greedy targeted attack. Finally, it is shown that {\color{black}under} our fragility measure {\color{black} structures are more brittle to edge removal than much of the previous literature would have indicated}, as demonstrated with real and natural networks. 
\end{abstract}

\maketitle


\section{Introduction}
Complex networks can be found in many areas of our lives, from the brain \cite{van2010,bassett2011,park2013,fish2021}, to our infrastructure \cite{braess2005,motter2013,torres2017,guimera2004} to our social interactions \cite{miller2015,himelboim2017} among others. Given how central they are, a question which may arise is how fragile/robust are these networks to lost edges or nodes? 
\paragraph*{}
As an example, when hurricane Irene ravaged the eastern coastline in $2011$, one region in northern New York was effectively cut off from aid due to the destruction of NY 73 \cite{mackenzie2016}. In Germany in 2006, a single high voltage power line was shut off to allow a cruise ship to pass, triggering a power outage for millions of people \cite{UCTE2006}.
\paragraph*{}
While the above situations were quite different, it is clear that an understanding of what makes a complex network robust to failure is necessary to avoid potentially catastrophic situations. In the present work we will focus on the  fragility of a graph to edge removal.
\paragraph*{}
It is not always evident that a network will be fragile to edge removal. In the cases above the fragility could be attributed to the sparsity of the network, that is the vast majority of possible edges available are not realized within the graph. However, even relatively dense networks may be quite fragile, a simple example of this is shown in Fig. \ref{fig:Fragile_example}. The removal of a single edge is capable of splitting the graph into two equal size components.
\begin{figure}[htbp] 
\includegraphics[width=0.35\textwidth]{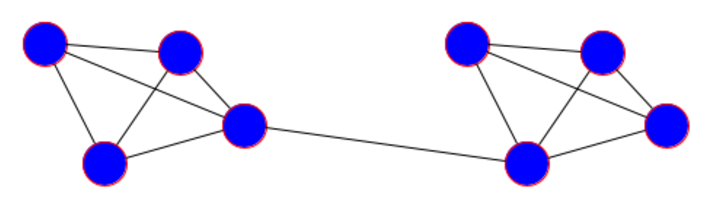}
\caption{An example of a graph that is very fragile to edge removal, despite being fairly dense.} \label{fig:Fragile_example}
\end{figure} 
The single edge connecting two highly connected components can be thought of as a bottleneck, and the graph thus has a small Cheeger constant \cite{chung1997,chung2005}.
\paragraph*{}
The example shown in Fig. \ref{fig:Fragile_example} gives the overall sense of what fragile means, but the term fragility remains fuzzy. Intuitively, if only a `few' edges must be removed to break the network into `small' pieces, then we will call a network fragile in terms of edge removal. In order to make clear what is meant by `few' and `small', the central question we will be asking in this work is the following: what fraction of edges must be removed for the resulting graph to have some fraction of nodes remain in the largest connected component? It is through this question we will make clear what we mean by the term fragility.
\paragraph*{}
{\color{black} There have been a variety of attempts to tackle the problem of fragility of networks.} For instance, fragility has been measured in the context of a percolation transition \cite{cohen2000,cohen2001}, where a network disintegrates almost surely under attack (random or targeted) after a critical percentage $p_c$ of edges or nodes have been removed. This idea of fragility is frequently used and has led to numerous questions about what makes a network more or less fragile. For instance it was found that networks which are more assortative appear to be more robust to random attack \cite{newman2003}. Attempts have been made under this same idea of fragility to optimize networks against random and targeted attacks simultaneously \cite{paul2004}.
\paragraph*{}
Other measures of fragility/robustness have been proposed as well, and generally are connected to the size of the largest connected component. One such measure proposed in \cite{schneider2011,louzada2015}:
\begin{equation}
   \mbox{R}_{\mbox{n}} = \frac{1}{n}\sum \limits_{\mbox{Q}=1}^{n} f(\mbox{Q}),
\end{equation}
where $f(\mbox{Q})$ represents the fraction of the $n$ nodes {\color{black} remaining} in the largest connected component after removing $\mbox{Q}$ nodes, {\color{black}averaged} over the fraction {\color{black}which remain} in the largest connected component for each {\color{black}$Q$ from $1$ to $n$}, and under this metric discovered `onion like' networks were most robust. This metric was also generalized to edge attacks \cite{zeng2012,duan2016}
\begin{equation}
     \mbox{R}_{\mbox{l}} = \frac{1}{m}\sum \limits_{\mbox{P}=1}^{m} f(\mbox{P})
\end{equation}
with $f(\mbox{P})$ denoting the fraction of nodes in the largest connected component after the removal of $\mbox{P}$ edges. 
{\color{black}It} was discovered that the optimal network in the case of node removals is not optimal in the case of edge attacks. 
\paragraph*{}
{\color{black} This problem and variants of it have also been posed in a variety of other ways. For instance, the minimum cut problem or more precisely, the minimum $k$ cut problem, which can be found with an algorithm which is polynomial in the number of nodes \cite{karger1996}, with high probability. However, as will be discussed in more detail below, we will place an additional constraint on the maximum size of the largest connected component which makes the problem more challenging.} 
\paragraph*{}
{\color{black} While not directly related, a similar problem is one of community detection. Several techniques have been developed for this as well such as the Girvan-Newman \cite{newman2004finding} algorithm, the Kernighan-Lin bisection \cite{kernighan1970} (which has been generalized to a k-section), the METIS algorithm \cite{karypis1997} (as a k-section method), and what has become known as the Louvain \cite{blondel2008} community detection algorithm among others.}
\paragraph*{}
In the existing literature, the fragility/robustness of a network to targeted attack is judged based on its ability to withstand sequential attacks, and generally view the fragility/robustness as a single measure.  However, this paradigm leads to an underestimate of the true fragility and ignores scenarios where one may wish to compare graphs based on their ability to absorb a certain level of destruction. Hence, below, we will introduce a measure which is not based on sequential attack, for which exact values can be obtained for certain graphs in the limit as the number of nodes goes to infinity, and is not a single measure per network. {\color{black} Additionally, we utilize several algorithms to provide estimates of the measure, which are tested on both randomly generated structures, as well as a real world example.}
\paragraph*{}
The remainder of this work is laid out as follows.
We will begin by offering a definition of {\it fragility}, relative to the fraction of edges which have been removed from the original graph (network). Next we define {\it robust} graphs, and follow with a proof that a certain class of graphs is robust. Finally, we will investigate how fragile some real networks are under the newly defined measure.
\section{Main Results}
\subsection{Network Fragility}
Throughout this paper, we will use the terms graph and network interchangeably. We begin by offering a definition of a graph.
\begin{definition}{Graph:}
A {\it graph} $G=(V,E)$ is a set containing $n$ vertices $V = \{v_1,v_2,...,v_n\}$ and a set of edges $E \subseteq (V \times V)$.
\end{definition}
Graphs can be broken into directed and undirected graphs, based upon whether the edges are oriented. 
\begin{definition}{Undirected graph:}
An {\it undirected} graph is a graph in which if $(v_i,v_j) \in E$, then $(v_j,v_i) \in E$.
\end{definition}
We will be exclusively exploring undirected graphs and $(v_i,v_j)$ along with $(v_j,v_i)$ will be counted as a single edge. As we are interested in how a graph changes with removed edges, we define a {\it perturbed graph} below.
\begin{definition}{Perturbed graph:}
A {\it perturbed} graph $G'(r) = (V,E')$ of $G$ is a graph with $r$ subtracted edges such that,
\begin{equation}
E \cap E' = E',
\end{equation}
and
\begin{equation}
\mbox{card}(E\backslash E') = r,
\end{equation}
\end{definition}
where $\mbox{card}(\cdot)$ is the cardinality of a set.
We call a graph {\it connected} if for every disjoint partition of nodes there is at least one edge which joins nodes from different partitions. Not all graphs are connected however, and thus we will define the {\it largest connected component} of a graph.
\begin{definition}{Largest connected component:}
Begin with a partition of $V$ into two disjoint sets $W,X \subset V$, with
\begin{equation}
W \cup X = V,
\end{equation}
\begin{equation}
W \cap X = \emptyset.
\end{equation}
Let 
\begin{equation}
E' = (W \times W) \cap E,
\end{equation}
and $G' = (W,E')$.
$G'$ is called a component of $G$. Furthermore, noting that a simple relabeling of the nodes of $G$ remains the same graph $G$, let $W = \{v_1,...,v_k\}$ and $X = \{v_{k+1},...,v_n\}.$ Finally, let,
\begin{equation}
C = \{x,y | (x,y) \in E'\},
\end{equation}
{\color{black} which is the set containing all unique nodes forming edges in $E'$}.
The {\it largest connected component} $\mbox{LCC}(G)$ of $G$ is the component $G'$ with largest $\mbox{card}(W)$ over all possible node relabelings, with,
\begin{equation}
v_i \in C \ (\forall i \in \{1,...,k\}).
\end{equation}
In other words, $\mbox{LCC}(G)$ is the largest possible component in which all nodes of the component have at least one edge.
\end{definition}
Certain classes of graphs are amenable to analysis because of their predictable structure. One such class of graphs are complete graphs {\color{black}and} will now be defined.
\begin{definition}{Complete Graph:}
A graph $G$ is called {\it complete} if the edge-set of the graph is given by $E = (V \times V) / \{(v_i,v_i)\}$, and the complete graph on $n$ nodes will be denoted by $K_n$. The number of edges of $K_n$ will be written as $\mbox{card}(E(K_n)) = \frac{n(n-1)}{2}$, where $E(\cdot)$ is the edge set of a graph.
\end{definition}
The central question that we will examine throughout this paper is how quickly a graph falls apart after targeted edge removals.  By fall apart, we mean that a component disintegrates to smaller components. Suppose one wants to know what the minimal number of edges required to be removed from $K_n$ in order to have a connected component of no larger than size $c$. As it turns out, the most effective strategy is to split the graph into as many components of size $c$ and a single component of size $b<c$ to remove the minimal number of edges. 

We now require the following result as a preliminary for what follows. 
\begin{theorem}{Splitting squares \label{thm:SplittingSquares}}
\\
Let $c> d_0 \geq d_1 \geq ... \geq d_l \geq d_{l+1}= 0$, $c,d_i \in \mathbb{N} \ (\forall i \in \{0,1,...,l+1\})$, $a_0 = c-d_0$, $a_i = d_{i-1}-d_i \ (\forall i >0)$ and $\sum \limits_{i=0}^{l+1} a_i = c$. Then 
\begin{equation}
c^2 \geq (c-d_0)^2 + \sum \limits_{i=1}^{l-1} (d_{i-1} -d_{i})^2 + d_l^2. \label{eq:SplittingSquares}
\end{equation}
In other words splitting $c$ into any number of integers will always result in a sum of squares which is less than or equal to $c^2$. \label{thm:splittingsquares}
\end{theorem}
\begin{proof}
It is clear that:
\begin{equation}
(c+d_0)(c-d_0) \geq (c-d_0)(c-d_0),
\end{equation}
since $c\geq d_i \geq 0$. Then
\begin{equation}
c^2 - d_0^2 \geq (c-d_0)^2 \implies c^2 \geq (c-d_0)^2 + d_0^2.\label{eq:Split1}
\end{equation}
Now suppose we split $c$ with another term. Then,  {\color{black} from (\ref{eq:Split1})}:
\begin{equation}
(c-d_0)^2 + (d_0+d_1)(d_0-d_1) \geq (c-d_0)^2 + (d_0-d_1)(d_0-d_1) \implies c^2 \geq (c-d_0)^2 + d_0^2 \geq (c-d_0)^2+(d_0 - d_1)^2 +d_1^2
\end{equation}
Finally suppose we split $c$, $(l+1)$ times, we have:
\begin{equation}
c^2 \geq (c-d_0)^2 + (d_0-d_1)^2 + ... + (d_{l-1}+d_{l})(d_{l-1}-d_{l}) \geq (c-d_0)^2 + (d_0-d_1)^2 + ... + (d_{l-1}-d_{l})(d_{l-1}-d_{l}),
\end{equation}
which can be rewritten as Eq. \ref{eq:SplittingSquares}.
\end{proof}
Exploiting Thm. \ref{thm:SplittingSquares}, the method outlined above can be shown to be the most efficient edge removal strategy for complete graphs.
\begin{theorem}{Efficient destruction of Complete graphs via edge removal \label{thm:Efficient1}}
\\
Let the graph $K_n$ be a complete graph, and let $\mbox{LCC}(K'(r)_n)$ be the largest component of the perturbed graph $K'(r)_n$. Set $r_{complete}^*(c) = \min \limits_r (\mbox{card}(\mbox{LCC}(K'(r)_n)) = c)$ with $c<n$. Then 
\begin{equation}
r_{complete}^*(c) = \mbox{card}(E(K_n)) - \Bigl( \Bigl \lfloor \frac{n}{c} \Bigr \rfloor \mbox{card}(E(K_c)) + \mbox{card}(E(K_{b})) \Bigr), \label{eq:CompGraphr}
\end{equation}
where $b = (n \bmod c)$ and $\lfloor \cdot \rfloor$ is the floor function.
\end{theorem}
\begin{proof}
Let $m = \mbox{card}(E(K_n))$, $b = (n \bmod c)$ with $c$ the number of nodes in the largest connected component, and the number of edges remaining after $r$ removals be $q$. Clearly 
\begin{equation}
r = m-q.
\end{equation}
Note that since $E(K_n) = \frac{n^2-n}{2}$ that $E(K_n)$ grows as $n^2$. Suppose that $b = 0$, in other words suppose that $c$ evenly divides $n$. Then $\frac{n}{c}\mbox{card}(E(K_c))$ is the largest possible number of edges remaining in the graph, and 
\begin{equation}
r^*(c) = m-\frac{n}{c}\mbox{card}(E(K_c)). \label{eq:Optimal1}
\end{equation}
Eq. \ref{eq:Optimal1} follows from Thm. \ref{thm:SplittingSquares} and from the fact that $\mbox{card}(E(K_n))$ grows as $n^2$. Now suppose that $b>0$ and note that $b<c$ and also note that $n = \Bigl \lfloor \frac{n}{c} \Bigr \rfloor c+ b$. Now both $b$ and $c$ can be split and combined in any manner to form a sum of squares which may be written:
\begin{equation}
e = \sum \limits_{i=0}^l a_i^2,
\end{equation}
with $\sum \limits_{i=0}^l a_i = \Bigl(\Bigl \lfloor \frac{n}{c} \Bigr \rfloor c+b \Bigr)$ and $0 \leq a_i \leq c \ (\forall i)$ since $c$ is the largest allowable connected component. We may now order the terms such that $\sum \limits_{i=0}^{l-k}a_i = \Bigl \lfloor \frac{n}{c} \Bigr \rfloor c$ and $\sum \limits_{l-k+1}^l a_i = b$, with $a_i \leq b \ (\forall i \geq [l-k+1])$ . Then 
\begin{equation}
e \leq b^2 + \Bigl \lfloor \frac{n}{c} \Bigr \rfloor c^2,
\end{equation}
follows from Thm. \ref{thm:SplittingSquares}. This implies that:
\begin{equation}
r^*(c) = m-\Bigl(\Bigl \lfloor \frac{n}{c}\Bigr \rfloor \mbox{card}(E(K_c))+ \mbox{card}(E(K_b))\Bigr) \label{eq:Optimal2}, 
\end{equation}
for the complete graph.

\end{proof}
We assert that the complete graph is the least fragile in terms of edge removals {\color{black} by virtue of the fact that any graph $G = (V,E)$, will be a subgraph of its corresponding (same size) complete graph $K = (V,E')$, with $E \subset E'$}. Thus in theory, in situations where one is concerned with the failure of a graph due to edge removals, one would design the graph of the system to be a complete one. However this is impractical in real world situations, often due to financial constraints as well as the enormity of the graphs at hand that tend to favor sparsity. For instance in the US, according to the Bureau of Transportation Statistics, there are approximately $20,000$ {\color{black} airports}, meaning that to connect every  {\color{black}airport} directly to every other, it would be necessary to build $\approx {\color{black}400,000,000}$ {\color{black}{flights}}. {\color{black}Performing all of these flights} would be prohibitively expensive. {\color{black}We note that many of the methods presented in the remainder of the paper would be prohibitively expensive for estimation of the fragility of such large graphs, with the exception of METIS which has its own limitations (for instance needing to form a balanced partition), thus in the remainder of the paper we examine much smaller graphs. We believe these smaller graphs illustrate the main point of the paper, however, that most methods underestimate the fragility of networks.} Instead, it is desirable to find graphs which are less dense, but which have a similar level of stability against edge removals. One such graph, which we will call the complete equitable bipartite graph, is presented in Fig. \ref{fig:CEBExamples}.
\begin{figure}[htbp] 
\includegraphics[width=0.5\textwidth]{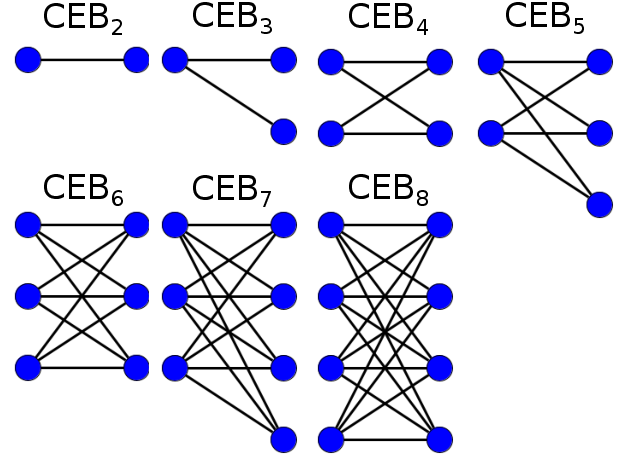}
\caption{Examples of CEB graphs. For even $n$ we can see that there are $\bigl(\frac{n}{2}\bigr)^2$ edges, while for odd $n$ the number of edges {\color{black}is given by} $\bigl(\frac{n+1}{2}\bigr)^2-\frac{n+1}{2}$.} \label{fig:CEBExamples}
\end{figure} 
\begin{definition}{Complete Equitable Bipartite (CEB) Graph:}
\\
We call a graph $G$ with number of nodes $n$ a {\it complete equitable bipartite} (CEB) graph if the graph is partitioned into two disjoint sets of nodes $W \cup X = V, W \cap X = \emptyset$ such that the cardinality of $W$ and $X$ differs by at most 1 (that is $|\mbox{card}(W)-\mbox{card}(X)|\leq 1$, where $|\cdot|$ is the absolute value and if the graph has edge-set $E = (W \times X) \cup (X \times W)$.
\end{definition}
The CEB graph on $n$ nodes will be denoted by $\mbox{CEB}_n$. Note that,
\begin{equation}
\mbox{card}(E(\mbox{CEB}_n)) = 
\begin{cases}
\Bigl(\frac{n}{2}\Bigr)^2, \ \ \mbox{for} \ n \ \mbox{even},
\\
\Bigl(\frac{n+1}{2}\Bigr)^2 - \frac{n+1}{2}, \ \ \mbox{for} \ n \ \mbox{odd},
\end{cases} 
\end{equation} 
as can be seen in Fig. \ref{fig:CEBExamples}.
Since $\mbox{card}(E(\mbox{CEB}_n)) \propto n^2$, we can follow the logic of Thm. \ref{thm:Efficient1} and find that: 
\begin{equation}
r^*_{\mbox{CEB}}(c) = \mbox{card}(E(\mbox{CEB}_n)) - \Bigl( \Bigl \lfloor \frac{n}{c} \Bigr \rfloor \mbox{card}(E(\mbox{CEB}_c)) + \mbox{card}(E(\mbox{CEB}_{b})) \Bigr).
\end{equation}
In order to make a comparison of graphs, we define the fragility of a graph to edge removals.
\begin{definition}{Fragility:}
\\
Let $G$ be a graph,  $\delta = \frac{c}{\mbox{card}(G)}<1$ be the fractional component size and $f_G(\delta) = \frac{r^*(c)}{\mbox{card}(E(G))}$ be the critical edge fraction. Then we define the {\it fragility} of the graph $G$ as:
\begin{equation}
\mathcal{F}_{\delta}(G) = 1-\frac{f_G(\delta)}{f_{\mbox{comp}}(\delta)},
\end{equation}
where, 
\begin{equation}
f_{\mbox{comp}} = \frac{ \mbox{card}(E(K_n)) - \Bigl( \Bigl \lfloor \frac{n}{c} \Bigr \rfloor \mbox{card}(E(K_c)) + \mbox{card}(E(K_{b})) \Bigr)}{ \mbox{card}(E(K_n))}
\end{equation}
is the critical edge fraction of the complete graph.
\end{definition}
{\color{black} We assert, without proof, that}, any graph $G$, $f_{\mbox{comp}}(\delta) \geq f_G(\delta) \ (\forall \delta)$, which if true means that $\mathcal{F}_\delta(G) \in [0,1]$. Now that the notion of fragility is defined, it is only natural to examine what it means for a graph to be robust.
\begin{definition}{Robust graphs:}
\\
We call a graph robust if for a given $\delta<1$, $\mathcal{F}_{\delta}(G) < \epsilon$, where {\color{black}$\epsilon \rightarrow 0$}. Additionally, we will call a graph asymptotically robust if $\forall \delta<1$, $\mathcal{F}_\delta(G) \rightarrow 0$ when $n \rightarrow \infty$. 
\end{definition}

Clearly, the complete graph is asymptotically robust. Now we will show that CEB graphs are robust as well.
\begin{theorem}{CEB Graphs Are Asymptotically Robust}
\\
If a graph is CEB then it is asymptotically robust.
\end{theorem}
\begin{proof}
Note that in the case of $n$ even and $c$ even we have:
\begin{equation}
\mathcal{F}_{\delta}(\mbox{CEB}_n) = 1 - \frac{\frac{n^2-n}{2}[(\frac{n}{2})^2 - (\lfloor \frac{n}{c}\rfloor (\frac{c}{2})^2 + (\frac{b}{2})^2)]}{(\frac{n}{2})^2 [(\frac{n^2-n}{2})-(\lfloor \frac{n}{c} \rfloor (\frac{c^2-c}{2})+(\frac{b^2-b}{2}))]}.
\label{eq:AllEven}
\end{equation}
For $n$ even and $c$ odd,
\begin{equation}
\mathcal{F}_{\delta}(\mbox{CEB}_n) = 1 - \frac{\frac{n^2-n}{2}[(\frac{n}{2})^2 - (\lfloor \frac{n}{c}\rfloor (\frac{c^2-1}{4}) + (\frac{b^2-1}{4}))]}{(\frac{n}{2})^2 [(\frac{n^2-n}{2})-(\lfloor \frac{n}{c} \rfloor (\frac{c^2-c}{2})+(\frac{b^2-b}{2}))]},
\label{eq:EvenThenOdd}
\end{equation}
in the case of $n$ odd and $c$ even,
\begin{equation}
\mathcal{F}_{\delta}(\mbox{CEB}_n) = 1 - \frac{\frac{n^2-n}{2}[(\frac{n^2-1}{4}) - (\lfloor \frac{n}{c}\rfloor (\frac{c}{2})^2 + (\frac{b^2-1}{4}))]}{(\frac{n^2-1}{4}) [(\frac{n^2-n}{2})-(\lfloor \frac{n}{c} \rfloor (\frac{c^2-c}{2})+(\frac{b^2-b}{2}))]},
\label{eq:OddThenEven}
\end{equation}
and finally for $n$ odd and $c$ odd,
\begin{equation}
\mathcal{F}_{\delta}(\mbox{CEB}_n) = 1 - \frac{\frac{n^2-n}{2}[(\frac{n^2-1}{4}) - (\lfloor \frac{n}{c}\rfloor (\frac{c^2-1}{4}) + (\frac{b}{2})^2)]}{(\frac{n^2-1}{4}) [(\frac{n^2-n}{2})-(\lfloor \frac{n}{c} \rfloor (\frac{c^2-c}{2})+(\frac{b^2-b}{2}))]}.
\label{eq:allEven}
\end{equation}
Examining the case of $n$ even and $c = n-1$ and noting that for $b = 1$ the term containing $b$ is $0$,
\begin{equation}
\mathcal{F}_{\frac{n-1}{n}}(\mbox{CEB}_n) = 1 - \frac{\frac{n^2-n}{2}[(\frac{n}{2})^2 -  ((\frac{(n-1)^2-1}{4}) )]}{(\frac{(n-1)^2-1}{4})^2 [(\frac{n^2-n}{2})-((\frac{(n-1)^2-(n-1)}{2}))]}.
\label{eq:NminusOne}
\end{equation}
Taking the limit of Eq. \ref{eq:NminusOne} as $n\rightarrow \infty$ and noting that the largest terms in both the numerator and denominator are $\frac{1}{8}n^4$ it is easy to see that in this case $\lim_{n \to \infty} \mathcal{F}_\frac{n-1}{n} (\mbox{CEB}_n) = 0.$ For $n$ even and $c = n-2$  it can be seen that :
\begin{equation}
\mathcal{F}_{\frac{n-2}{n}}(\mbox{CEB}_n) = 1 - \frac{\frac{n^2-n}{2}[(\frac{n}{2})^2 -  ((\frac{(n-2)}{2})^2 +1 )]}{(\frac{n}{2})^2 [(\frac{n^2-n}{2})-((\frac{(n-2)^2-(n-2)}{2})+1)]}.
\label{eq:NminusTwo}
\end{equation}
Eq. \ref{eq:NminusTwo} again leads to $\mathcal{F}_{\frac{n-2}{n}}(\mbox{CEB}_n)=0 $ in the limit as $n \rightarrow \infty$. In general for $k < \frac{n}{2}$ we have:
\begin{equation}
\mathcal{F}_{\frac{n-k}{n}}(\mbox{CEB}_n) =
\begin{cases}
1 - \frac{\frac{n^2-n}{2}[(\frac{n}{2})^2 -  ((\frac{(n-k)}{2})^2 +(\frac{k}{2})^2 )]}{(\frac{n}{2})^2 [(\frac{n^2-n}{2})-((\frac{(n-k)^2-(n-k)}{2})+(\frac{k^2-k}{2}))]}, & \mbox{for k even} \\
1 - \frac{\frac{n^2-n}{2}[(\frac{n}{2})^2 -  ((\frac{(n-k)^2-1}{4})+(\frac{k^2-1}{4}) )]}{(\frac{(n-k)^2-1}{4})^2 [(\frac{n^2-n}{2})-((\frac{(n-k)^2-(n-k)}{2})+(\frac{k^2-k}{2}))]} & \mbox{for k odd}.
\end{cases}
\label{eq:Nminusk}
\end{equation}
Eq.  \ref{eq:Nminusk} approaches $0$ as $n \rightarrow \infty$ for all $k$. Now since $n$ is assumed even, the case of $k=n/2$ will be examined, this leads to:
\begin{equation}
\mathcal{F}_{0.5}(\mbox{CEB}_n) = \frac{\frac{n^2-n}{2}[(\frac{n}{2})^2 - 2(\frac{n}{4})^2]}{(\frac{n}{2})^2[\frac{n^2-n}{2} - 2\frac{(\frac{n}{2})^2-\frac{n}{2}}{2}]}, \label{eq:NminusNoverTwo}
\end{equation}
so we find that $\mathcal{F}_{0.5}(\mbox{CEB}_n) = 0$ as $n \rightarrow \infty$. Since the only thing that changes for increasing $k$ beyond this point is the pre-factor in front of the third terms in both the numerator and denominator, it is clear that for all $k$, $\mathcal{F}_{\frac{n-k}{n}}(\mbox{CEB}_n) = 0$ as $n \rightarrow \infty$. This completes the proof for the case of $n$ even. The proof follows similarly when $n$ is odd.  
\end{proof}
\begin{figure}[htbp] 
\includegraphics[width=0.65\textwidth]{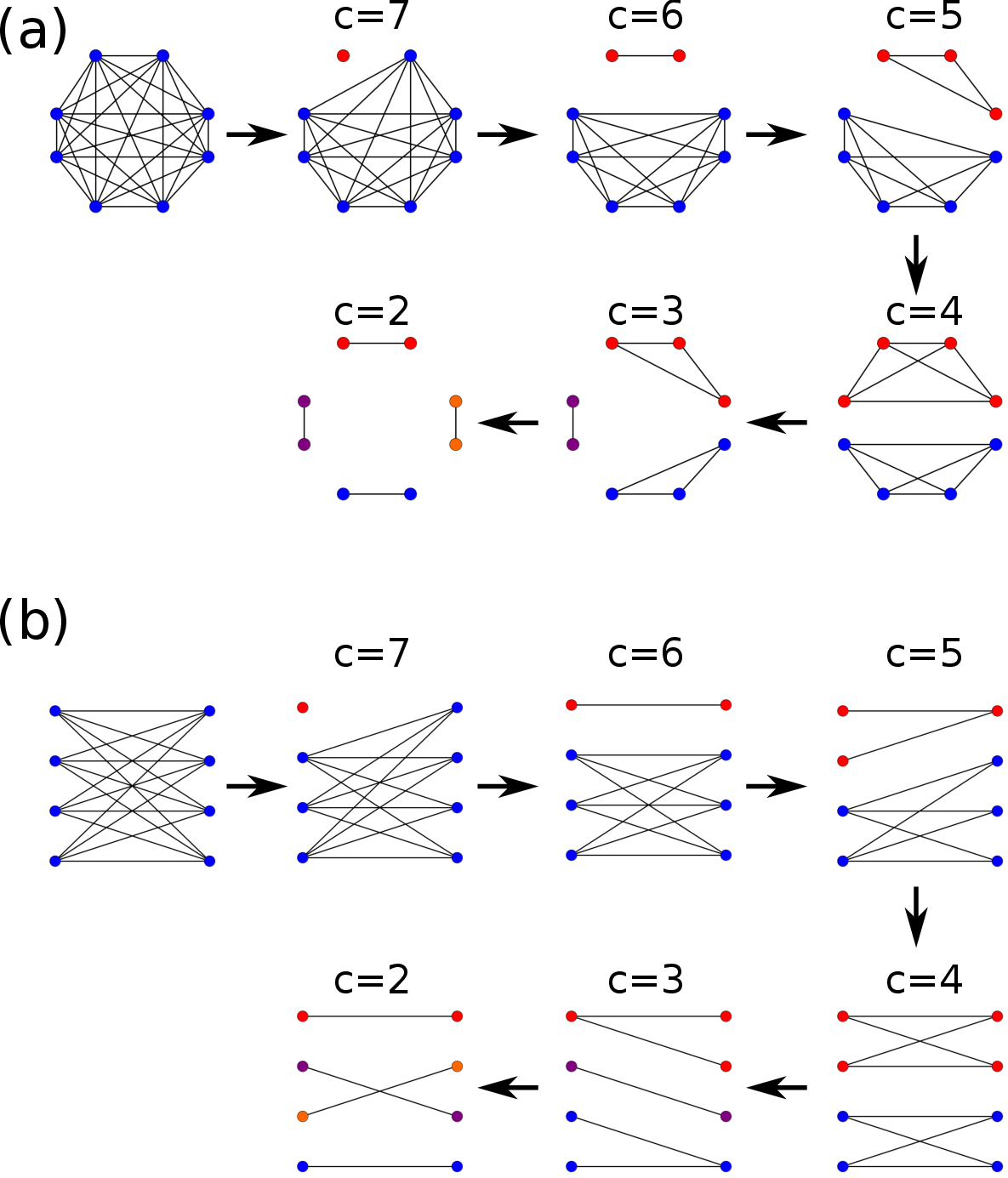}
\caption{Minimal edge removal for maximum component size $c$ in (a) a complete graph on 8 nodes ($K_8$) and (b) a CEB graph on 8 nodes ($\mbox{CEB}_8$).} \label{fig:CEB}
\end{figure} 

Of note is that this is not true in general for graphs with the same number of edges as the CEB. For instance, consider the case of even $n$ with two complete graphs of size $\frac{n}{2}$ connected together by $n$ edges, which we will call a generalized barbell or GB graph. Such a graph has the same number of edges as the CEB, and yet we can see that $\mathcal{F}_{0.5}(\mbox{GB}_n) = 1$ as $n \rightarrow \infty$ since the number of edges of the GB graph grows as order of $n^2$ but the number of edges required to split the GB graph in half grows as $n$, as opposed to the complete graph in which both the number of edges and the number of edges required to split it in half grows as order of $n^2$.
{\color{black}As is clear from the example of the generalized barbell, and prior work \cite{shi2024local}, community structure will play a large role in $\mathcal{F}_\delta(G)$, depending on the size of the communities and the values of $\delta$. The generalized barbell has two communities which are loosely connected to each other, which is the source of its fragility at $\delta=0.5$. This generalizes to a higher number of communities as well, for instance if three equal sized complete graphs are connected to each other with few edges, the fragility will be large at $\delta=\frac{1}{3}$, and so on. While real world networks are not as well structured, the community structure will resemble the preceding examples, causing the fragility to be large when $\delta = \frac{1}{\rho}$, where $\rho$ is the number of communities. }
\paragraph*{}
An example of efficient destruction of both the complete graph (Fig. \ref{fig:CEB}(a)) and CEB graphs  (Fig. \ref{fig:CEB}(b)) is shown in Fig. \ref{fig:CEB}. It can be seen that qualitatively both types of graphs fall apart at the same rate. 

\subsection{Methods for Estimating Fragility}
In certain instances, as was the case for the $\mbox{CEB}_n$ graph, it is possible to obtain a closed form expression for the fragility of a graph. However except for special cases, such an expression may be unknown or not exist as is typical for graphs found in the real world. For this reason, $\mathcal{F}_\delta$ must be estimated. {\color{black} As can be seen in Fig. \ref{fig:MethodsNotGeneral}, simple methods for estimating the fragility may fail, even in scenarios that seem rather simple. For all three of the networks presented, bi-partitioning (corresponding to $\delta=0.5$) the graph can be accomplished by the removal of a single edge, so each of these graphs has a fragility near 1. However, different simple attack strategies may lead to a vast underestimate of the fragility. In (a) choosing to attack the edge that is connected to the highest degree nodes, or alternatively attacking the edge with highest degree sum (accounting for both nodes across the edge) would lead to determining the correct fragility, while attacking the lowest degree nodes would not. In (b) attacking the edges attached to the minimum degree nodes would lead to the correct fragility, while the other mentioned methods would not, and in (c) all of the mentioned methods of attack would fail to find the true fragility.} In this section, we outline {\color{black} } greedy {\color{black} methods} for the estimation of  $\mathcal{F}_\delta$. 
\begin{figure}[b]
\includegraphics[width=0.95\textwidth]{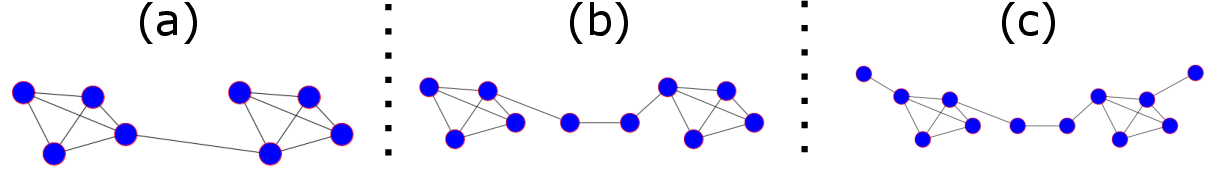}
\caption{\color{black} Shown are three networks which all have fragility close to 1 for $\delta=0.5$ since it only requires the removal of a single edge to bi-partition the graph. (a) Attacking minimum degree nodes will not find the true fragility here, though either attacking edges attached to maximum degree nodes or maximum edge-sum will. (b) Attacking edges attached to minimum degree nodes will work here, but the other methods discussed will not. (c) None of the methods will find the true fragility in this case. }
\label{fig:MethodsNotGeneral}
\end{figure}
In most prior work \cite{cohen2000,newman2003,paul2004,schneider2011,louzada2015,zeng2012,duan2016,ma2016,cohen2001,liu2017,peng2025unveiling}, the fragility of a network was estimated by greedy removal of either edges or nodes as given in Algorithm \ref{alg:GreedyEdge}. For edge removal, it is typical to apply a `destruction function', $f(\cdot)$, to each edge and choose the edge which maximizes the destruction of the network. However, recently, there has been a realization \cite{ren2018} that this may not be an optimal attack strategy. In other words networks may be more fragile than previously thought. For our purposes, we measure the amount of destruction by the size of the LCC after the edge has been removed. A smaller LCC implies a large value for $f$. A typical metric used to determine which edge to remove at each step is the edge betweenness. Edges with high edge betweenness are generally thought of as being of high importance to the network. Thus the value of $f$ in this case is the edge betweenness. An alternative attack strategy will also be used in this work, one related to the minimum degree node. In this case, every edge attached to the node of minimum degree in the network will have the same value of $f$, while edges for higher degree nodes have smaller values of $f$. Therefore the edges of the minimum degree node will be attacked first, until all such edges are stripped away. 
\begin{algorithm}[H]
\SetAlgoLined
\KwData{$G(V,E)$: Graph with vertices $V$ and edge set on $n$ edges, $E$, where $E = \{ e_{1}, e_{2}, \dots, e_{n} \}$}
\phantom{Data: } $r$: Number of removals 
\phantom{Data: } $f(\cdot)$: The "destruction function"
\KwResult{$G'$: Reduced graph}
\textbf{Initialization:} set $E' = E$ \\
 \For{$l = 1: r$}{
 \For{$k = 1:n-l+1$}{
  $a_{k} = f(e_{k}), \qquad e_{k} \in E'$
  }
  $b = \argmax_{k} a_{k}$ 
  $E' = E'\backslash e_{b}$
 }
 \textbf{Return:} $G'(V,E')$ 
\caption{Girvan-Newman($G(V,E), r, f(\cdot)$)}
\label{alg:GreedyEdge}
\end{algorithm}
Using greedy removal is a computationally efficient method to search for a set of edges $R \subseteq E$ to be removed from the edge set $E$ of the graph. We must resort to such a strategy in our search because the number possible edge sets for removal grows as $r!$, where $r = \mbox{card}(R).$ However greedy algorithms are known to produce sub-optimal results in certain circumstances, particularly if they are applied without corrective steps \cite{heller2017,jiang2014,zahedinejad2014}. 

\begin{figure}[b]
\includegraphics[width=0.49\textwidth]{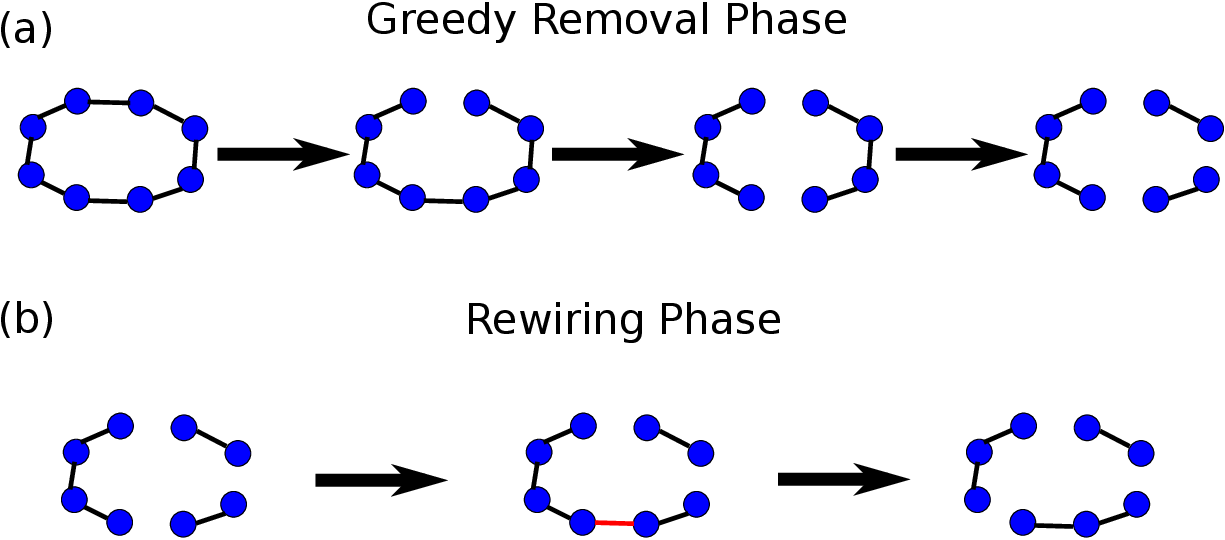}
\caption{\label{fig:fragility} Greedy Removal with Rewiring. Here we show how the algorithm works on an 8 node cycle with 3 edges to be removed. First in (a), greedy removal is performed. Then in (b), the algorithm attempts to rewire edges out of the LCC. If the rewiring results in a reduction in the size of the LCC, then it is accepted. Notice that the red edge in (b) existed originally in the network. Thus, the rewiring is constrained by the original network topology.}
\end{figure}
\begin{algorithm}[H]
\SetAlgoLined
\KwData{$G(V,E)$: Graph with $n$ vertices $V$ and edge set on $m$ edges, $E$, where $E = \{ e_{1}, e_{2}, \dots, e_{m} \}$, $V = \{v_{1},v_{2},...,v_{n}\}$}

\phantom{Data: } $r$: Number of removals 

\phantom{Data: } $f(\cdot)$: The destruction function

\KwResult{$G'$: Final graph}

\textbf{Initialization:} set $G'(V,E') =$ Girvan-Newman($G(V,E),r,f(\cdot)$), $\mbox{LCC}^{(1)}(V',E'') = \mbox{LCC}(G')$, where $E'' = \{e_{k_1}'', e''_{k_2}, ... ,e''_{k_a}\} \subset E$, $V' = \{v'_{p_1},v'_{p_2},...,v'_{p_l}\} \subset V$ \\
 
 \For{$i = 1: l$}{
 $S^{(1)} = \{e_j |e_j \in [(v'_{p_i} \times V) \cup (V \times v'_{p_i})] \cap E\}$ \\
 $S^{(2)} = \{e_j |e_j \in [(v'_{p_i} \times V') \cup (V' \times v'_{p_i})] \cap E''\}$ \\
 \If{$\mbox{card}(S^{(2)})\leq \frac{\mbox{card}(S^{(1)})}{2}$}{
 $S^{(3)} = S^{(1)}\backslash S^{(2)}$ \\
 Choose $E^{(1)} \subset S^{(3)}$ at random such that $card(E^{(1)}) = card(S^{(2)})$ \\
 $E^{(2)} = (E'' \backslash S^{(2)}) \cup E^{(1)}$ \\
 $G^{(1)} = (V,E^{(2)})$ \\
 $\mbox{LCC}^{(2)} = \mbox{LCC}(G^{(1)})$ \\
 \If{$\mbox{card}(\mbox{LCC}^{(2)}) < \mbox{card}(\mbox{LCC}^{(1)})$}{
 $G' = G^{(1)}$
 }
 }
 }
 \textbf{Return:} $G'(V,E')$ 
\caption{RewiringRemoval($G(V,E), r, f(\cdot)$)}
\label{alg:Rewiring}
\end{algorithm}

%
%

To better estimate the fragility of a network, we must move beyond a simple greedy algorithm. For this purpose, we begin by using two greedy removal strategies as the first stage: one chooses the edge with largest edge betweenness inside the LCC at each step, the other targets the edges of the lowest degree node in the LCC. In the early stages, the fastest way to destroy a network frequently is to attack the minimum degree nodes, though this is not always the case. However in later stages, especially for $\delta <<1$, the edge betweenness strategy is the most effective targeted attack strategy. Thus, combining these two, a more optimal set of edge removals may be obtained for any given $G$ and $\delta$.
\paragraph*{}
After the greedy removal phase has been completed, there is a perturbed network $G'(r)$ with edge set $E'$. A second stage of the algorithm is now performed, which involves rewiring edges from the $\mbox{LCC}^{(1)} = \mbox{LCC}(G'(r))$ to components outside of $\mbox{LCC}^{(1)}$. The rewiring is constrained by the original network structure as shown in Algorithm \ref{alg:Rewiring}. In this stage, candidate nodes are identified from $\mbox{LCC}^{(1)}$, with recognition that a node cannot be rewired out of $\mbox{LCC}^{(1)}$ if it has more edges inside $\mbox{LCC}^{(1)}$ than edges which have been removed from that node. Thus, only a subset of nodes in $\mbox{LCC}^{(1)}$ is chosen for the attempted rewiring. Once this subset has been determined, nodes from $\mbox{LCC}^{(1)}$ are rewired to other components of the network, and edges can only be swapped out for edges which were removed from $G$.
\paragraph*{}
The rewiring algorithm also faces a combinatorial problem. To see that this is the case, let $E''$ be the edges of $\mbox{LCC}^{(1)}$. Now define 


\begin{align}
\nonumber G &= (V,E), \\ \nonumber
G' &= (V,E'), \\ \nonumber
\mbox{LCC}^{(1)} &= (V',E''), \\ \nonumber
V' &= \{v_{p_1}, ..., v_{p_l}\} \\ 
E &= \{e_1,...,e_m\}\\ \nonumber
S_i^{(1)} &= \{e_j |e_j \in [(v'_{p_i} \times V) \cup (V \times v'_{p_i})] \cap E\}, \\ \nonumber
S_i^{(2)} &= \{e_j |e_j \in [(v'_{p_i} \times V') \cup (V' \times v'_{p_i})] \cap E''\} 
\label{eqn:rewiring}
\end{align}

Clearly, if 
\begin{equation}
    \frac{\mbox{card}(S^{(1)}_i)}{2} - \mbox{card}(S^{(2)}_i) > 0,
\end{equation}
then there will be choices of which edges to use from the original edge set for rewiring. When this issue arises, a single random set of edges with cardinality $\mbox{card}(S^{(2)}_i)$ is chosen.
\paragraph*{}
The rewiring is only accepted if the size of $\mbox{LCC}(G'(r))$ decreases. That is, if we let $\mbox{LCC}^{(2)}$ be the largest connected component of the rewired graph $G''(r)$, then rewiring is only performed if
\begin{equation}
    \mbox{card}(\mbox{LCC}^{(2)})<\mbox{card}(\mbox{LCC}^{(1)}).
\end{equation}  
It is possible that the rewired graph may allow for additional rewiring, so the algorithm is applied recursively until the largest component no longer decreases in size. This algorithm will never do worse than greedy removal, {\color{black} since the rewiring is only selected if it is better than the greedy removal}. Once the rewiring stage is completed for both greedy strategies a final stage is completed, as described below. The candidate removal set with the fewest edges removed is then chosen among the two candidate sets, one from the minimum degree attack strategy and the other from edge betweenness. 

\begin{algorithm}[H]
\SetAlgoLined
\KwData{$G(V,E)$: Graph with vertices $V$ and edge set on $n$ edges, $E$, where $E = \{ e_{1}, e_{2}, \dots, e_{n} \}$}

\phantom{Data: } $r$: Number of removals 

\phantom{Data: } $f(\cdot)$: The destruction function

\phantom{Data: } $c$: The maximum allowed $card(LCC)$

\KwResult{$G'$: Final graph}

\textbf{Initialization:} set $G'(V,E') =$ RewiringRemoval($G(V,E),r,f(\cdot)$), $n = \mbox{card}(E)$. \\
 
 \For{$l = 1: n$}{
 \If{$\mbox{card}(LCC(G' \cup e_l)) \leq c $}{
  $E' = E' \cup e_{l}$
  }
 }
 \textbf{Return:} $G'(V,E') $
\caption{IterativeAddBack($G(V,E), r, f(\cdot),c$)}
\label{alg:IterativeAddBack}
\end{algorithm}

Algorithm \ref{alg:IterativeAddBack} acts as the final step in the new algorithm for estimating the fragility of the network. This stage is performed after the greedy removal and rewiring stages have been completed. It involves iteratively adding back any edges from the original network $G$ to the perturbed graph $G'$ which do not increase the size of the $\mbox{LCC(G')}$ beyond the largest allowable component size $c$. This final stage allows components (typically other than the $\mbox{LCC}$ but not necessarily after rewiring) to be `regrown' up to cardinality $c$. {\color{black} We note that this method remains an iterative removal method, however, it improves upon the simply applied greedy removal methods available by adding edges back in that did not need to be removed to split the network into the desired pieces.} Code for this method is made available at \cite{fish2022}.

\subsection{Numerical Results}
In this section, a comparison of the various methods outlined above will be presented. Random edge removal and targeted attack are performed, both on synthetic as well as real networks. Performance of these attack methods is examined in terms of the network fragility measure developed above.

\subsubsection{Real Network Data}
\label{ssec:data}
A real network is generated from Safegraph data for comparison of the performance of the techniques.
Data was obtained from the Safegraph mobility dataset \cite{safegraph}. Location data was collected from over 20 million devices. We first consider a shopping mall, where each business has an independent entry (see Figure \ref{fig:WBMall} for layout). 

\begin{figure}[b]
\includegraphics[width=0.48\textwidth]{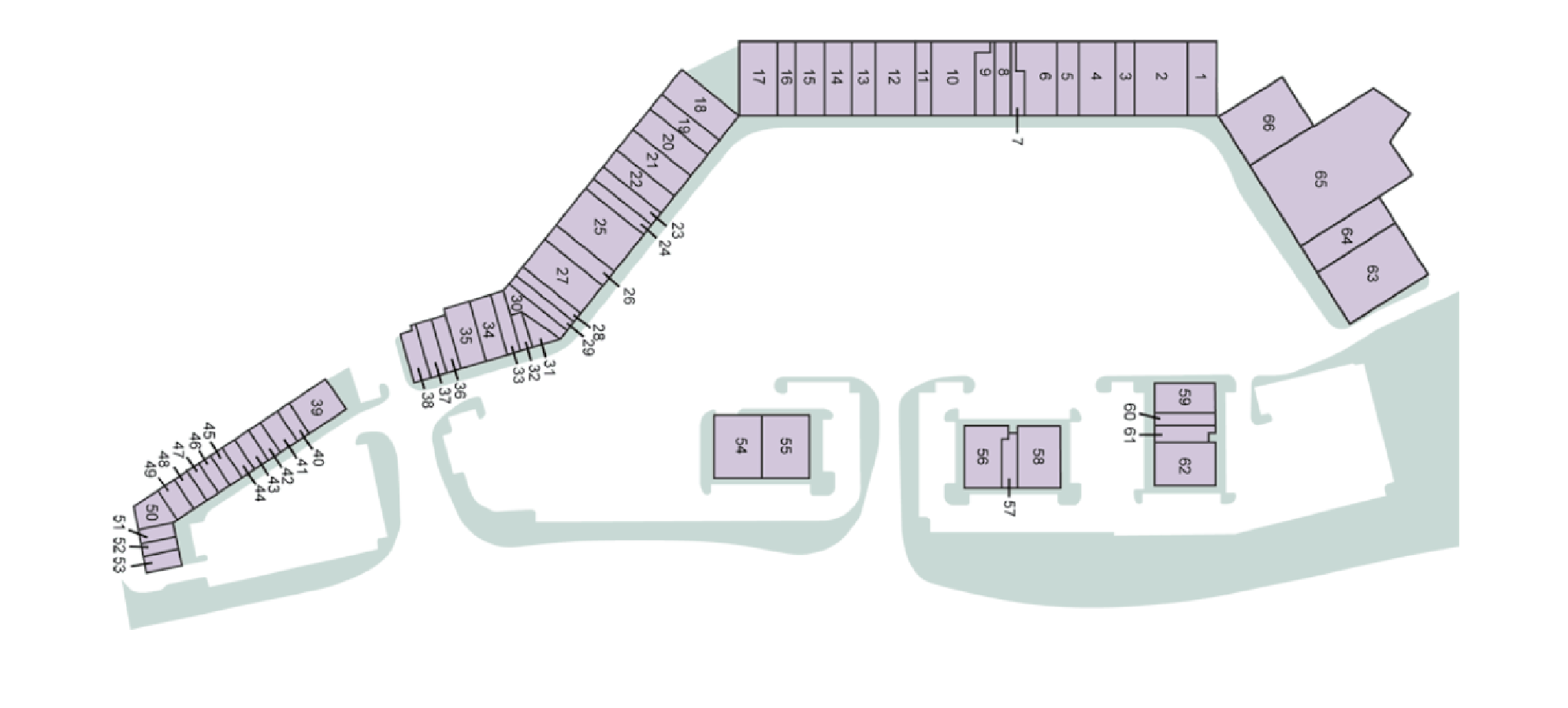}
\caption{Layout of the mall. Each numbered space is a business in the shopping mall. All businesses have outside entries, and are not connected to each other internally. }
\label{fig:WBMall}
\end{figure}

The Safegraph-tracked devices in the shopping mall were classified as (a) those whose location is precise to the business they are in; and (b) those devices are in the mall, but their location in the mall is not known more precisely. These data are sampled at hourly intervals. 

To build a network, we combine the Safegraph data with a publicly available layout of the mall. We start with a single snapshot in time. Devices whose locations are known at the business level were first placed. The rest of the devices were placed at random locations in the mall. A network is formed assuming Bluetooth connectivity between devices. Nominally, a ten-meter range is assumed for Bluetooth communications. A pair of devices within this range is assumed to be connected, unless there are walls between them. For each wall between the pair of devices, the range is halved \cite{Achalla20}. An example of a network in this way is shown in Figure \ref{fig:WBMall_4pm}, where the red dots are devices identified to be in specific businesses, and the blue dots, the devices placed at random locations. 

\begin{figure}[b]
\includegraphics[width=0.48\textwidth]{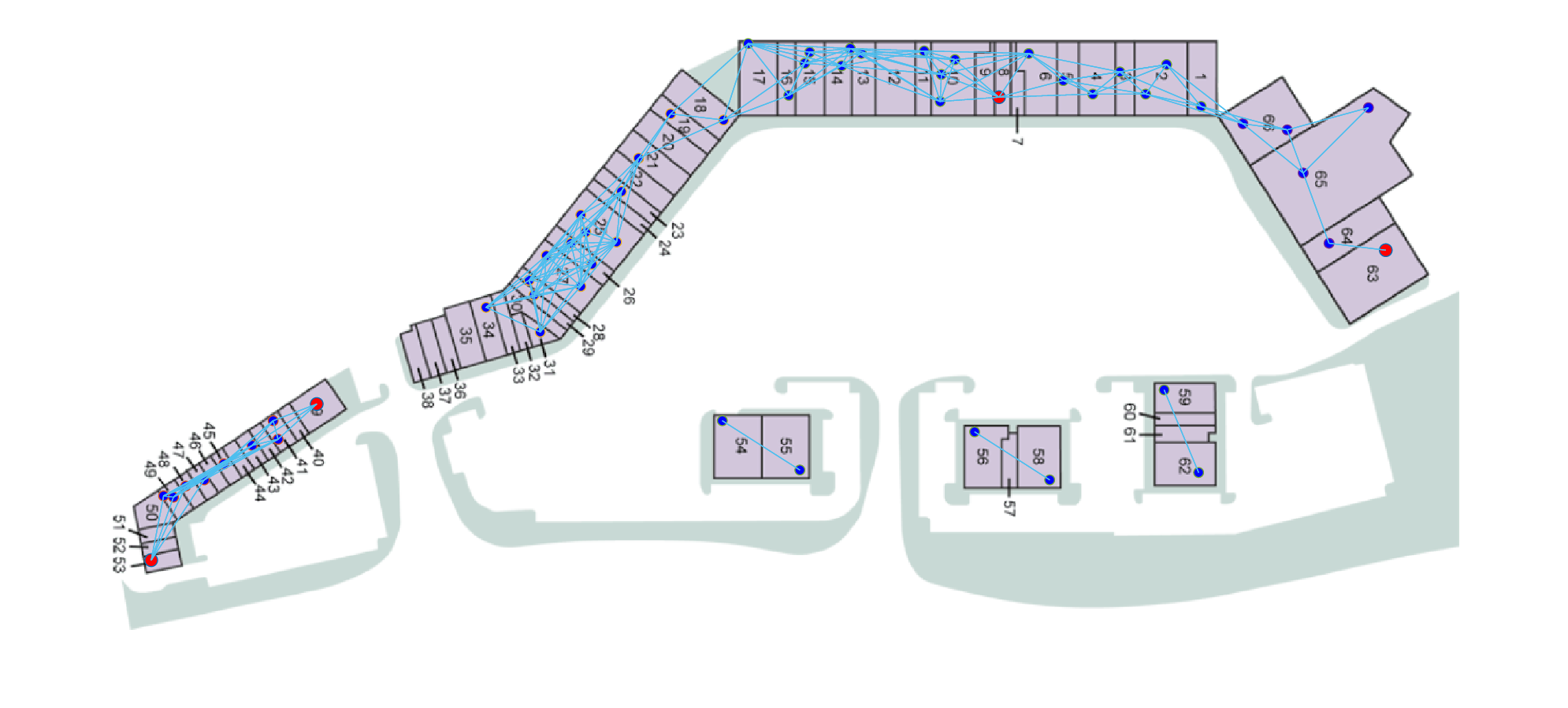}
\caption{Devices in the mall are connected to form a network. A pair of devices are connected if they are within communications range of each other. Red dots are devices identified to be in specific businesses, and the blue dots, the devices placed at random locations.}
\label{fig:WBMall_4pm}
\end{figure}

\subsubsection{Random Edge Removal}
\label{ssec:results_random_removal}

We consider three random graph architectures to assess the effects of random edge removal: the Watts-Strogatz (WS) {\color{black} \cite{watts1998collective} model, which initializes a ring of nodes connected to their $k$ nearest neighbors and then rewires edges drawn uniformly at random with probability $p$, the Erd\H{o}s-R\'{e}nyi (ER) \cite{erdds1959random} model, which has a parameter $p$ which controls the probability of any pair of nodes being connected by an edges, and the Barab\'{a}si-Albert (BA) \cite{barabasi1999emergence} model, which is a growing network model that has a parameter $m$ that controls the number of edges a node joining the network attaches to, and the probability of attachment to existing nodes is based upon the fraction of edges they are already attached to.} In addition, we consider the Safegraph mall network described above. For the synthetic graphs for each random realization, the degree distribution was determined. Edges were targeted at random, and after each edge was removed, the degree distribution of the resulting graph was evaluated. The process was continued until all edges were removed. The Hellinger divergence between the original graph degree distribution and the degree distribution after each removal was calculated.{\color{black} The Hellinger divegence can be calculated as follows:
Let $P$ and $Q$ be discrete valued probability distributions. Then the Hellinger divergence \cite{hellinger1909neue}, $H(P,Q)$, is defined as:
\begin{equation}
H(P,Q) = \frac{1}{\sqrt{2}}||P-Q||_2.
\end{equation}
For the discussion that follows, to find the distributions $P,Q$ we used the empirical probability for each degree in the degree distribution.
}
This entire process was repeated multiple times for each graph configuration and the Hellinger divergence values were averaged across 100 trials for each number of removed edges. Videos were produced, which are included as supplementary material, of the change both in network structure and Helinger divergence as each edge is removed. 
Random removal is not an efficient attack mechanism however and thus gives poor estimates of the fragility of a network. Indeed, in all cases examined the estimated value of $\mathcal{F}_{0.5}$ was negative, which clearly makes random removal inappropriate for this estimation.

\subsubsection{Targeted Attack}
We examined {\color{black} five} types of targeted attack, attacking the minimum degree nodes, attacking sequentially via edge-betweenness centrality, {\color{black} edge-sum which accounts for the total sum of the degrees across the edge, and edge-collective influence (ECI) \cite{peng2025unveiling} which is a generalization of collective influence \cite{morone2015influence} to edges, } and the newly proposed method with {\color{black} both rewiring and the iterative add back algorithm performed. We note that the scores for each of these methods is updated iteratively at each step, rather than in a single shot, to give each the best chance at estimating the fragility}. For comparison with the random edge removal method, the targeted attack methods were applied to the same networks as random removal, until $E = \emptyset$. Videos, included as supplementary material, have been produced showing changes in the degree distribution as well as the Hellinger divergence \cite{pardo2018} between the initial degree distribution and the new degree distribution (after each edge removal) for the various attack strategies.
\paragraph*{}
In  Table \ref{tab:RandomNets}, the fragility at $\delta = 0.5$ is estimated (i.e. the largest connected component is no larger than half of all of the nodes) by averaging over $100$ realizations of the BA, ER and WS graphs, each with $500$ nodes and exactly $1984$ edges. For the BA network the model parameters were $n=500$ and $m = 4$, for the ER graph $n = 500$, and $p$ is chosen for each network realization so that the number of edges is exactly $1984$, and for the WS network, $n = 500$ and $k = 8$, $p = 0.2$, and then edges were removed at random until the network had exactly $1984$ edges {\color{black} to maintain a constant edge density}. {\color{black} Additionally, } the fragility at $\delta = 0.5$ for the Safegraph mall network shown in Fig. \ref{fig:WBMall}{\color{black}, as well as three networks found on the network repository \cite{nr} US Air 97 (a flight network), DM-LC (a fruit fly protein-protein network) and the Florida ecological network, was estimated. 
These real world networks allow for a variety of implications in terms of their fragility. For instance, the Safegraph Mall network would have relevance to disease spread, as limiting contact between individuals who are sick would be preferential, and the estimated fragility of this network implies that breaking the network apart would be relatively easy in a control scenario, as relatively few edges are needed to be removed to split the graph in half. A similar story emerges for the US Air 97 network, though this seems to be undesirable, as one would easily be trapped in only one half of the network if relatively few flights were canceled. The high fragility of the fruit fly protein-protein network, implies a brittle genome for the fruit fly, meaning that small changes in the genome can lead to large effects, a fact that has been observed, see for example \cite{torres2025heterochrony}. Contrasting these with the Florida ecological network, a food web, which has a much lower fragility than the others, which implies the food network could absorb a larger impact at the species level before falling apart, which is clearly desirable to animals in the food chain}. {\color{black} Furthermore, we tested two other algorithms for bi-partitioning the graph, METIS \cite{karypis1997} and a modified version of the Kernighan-Lin algorithm \cite{kernighan1970} where the initial partition is not simply random. In this case, we initialize one partition to be the first half of the nodes in a depth first search rooted at the highest degree node (or if there is a tie from a randomly chosen node with highest degree).}
{\color{black} We observed that across the tested network types, the edge-sum algorithm consistently returned a lower estimate than both the minimum degree, and the edge betweenness value, and thus to preserve space in the table we did not report its value. The table clearly shows that estimation techniques such as edge betweenness, minimum degree, and edge collective influence consistently underestimate the fragility of the networks, in fact in none of the cases examined did any of those algorithms find the optimal removal set.}
\paragraph*{}
{\color{black} In Figure \ref{fig:FragsParams}, we see how the estimated fragility of networks changes as the density changes over networks containing $500$ nodes. In Fig. \ref{fig:FragsParams} (a) The Watts-Strogatz graphs, which are also known as small-world graphs, are shown for increasing $k$. As $k$ increases, the density of the network increases, since $k$ controls how many nearest neighbors are connected within a ring structure. As expected, for small values of $k$, the WS graphs are highly fragile, this is because it only requires a small set of edges to be removed to split a ring in half. The rewiring probabilities were chosen from the set $p \in \{0.1,0.3,0.5,0.7,0.9\}$, and when the rewiring probability is smaller, more of the ring structure is maintained, alternatively as the rewiring probability increases the network becomes more and more like an ER network. For smaller values of $p$ the estimated fragility remains higher longer, this is related to maintainence of the ring-like structure. However, interestingly, at higher densities, the networks with larger values of $p$ exhibit a non-monotonic behavior with increasing density, it is unclear if this is an artifact of the estimation techniques, or if there is some complex interplay between the density and the rewiring probability. In any case, for all values of $p$ the estimated fragility converges to $0$ as the networks get arbitrarily close the complete graph. In (b) for BA networks with increasing $m$, and (c) for ER networks with increasing $p$, we see that both network types do exhibit a monotonic behavior in fragility as the density increases. Of the tested network types, the ER networks appear to be the least fragile when holding density constant.}
\begin{table}[]

    \centering
    \begin{tabular}{|l|c|c|c|c|c|c|c|}
    \hline
       {\bf Network Type}  &{\bf Best Estimate ($ <\mathcal{F}_{0.5}>$)}& {\bf With rewiring} & {\bf METIS} & {\bf Kernighan-Lin} & {\bf Min. Degree} & {\bf Edge Betweennness} & {\bf ECI} \\ 
       \hline
        ER 
       & 0.66& 0.43 & {\bf 0.66} & 0.65& -0.03 & -0.01 & 0.04\\
       \hline      
       BA 
       & 0.64 & 0.36 & {\bf 0.64} & 0.63 & 0.05 & -0.07 & 0.23 \\
       \hline
       WS 
       & 0.74 & {\bf 0.74} & 0.68 & 0.68 & -0.10 & 0.58 & 0.54\\
       \hline
       Safegraph Mall & 0.98 & 0.84 &{\bf 0.98} & 0.97 & 0.84 & 0.84 & 0.67\\
       \hline
       {\color{black} US Air 97} & 0.80 & {\bf 0.80} & 0.64 & 0.76 & 0.70 & 0.75 & 0.34 \\
       \hline
       {\color{black} DM-LC} & 0.98 & {\bf 0.98} & {\bf 0.98} & 0.91 & 0.73 & 0.98 & 0.89 \\
       \hline
       {\color{black} Florida Ecological} & 0.33 & 0.26 & {\bf 0.33} & -0.02 & -0.05 & -0.07 & -0.15 \\
       \hline
      
    \end{tabular}
    \caption{\label{tab:NF} Estimated Network Fragility. These are the estimated fragility values with $\delta=0.5$ for three different network types, Erd\H{o}s-R\'enyi (ER), Barabasi-Albert (BA) and Watts-Strogatz (WS), as well as the Safegraph mall network. The values are averaged over 100 network realizations, and in each run, the networks each had $500$ nodes and exactly $1984$ undirected edges, with the exception of the real world graphs, the Safegraph Mall, which has $133$ nodes and $730$ edges. The US Air 97, a flight network, has $332$ nodes and $2126$ edges. DM-LC, which is the Drosophila Melanogaster (common fruit fly) protein-protein interaction network, has $658$ nodes and $1129$ edges. Finally, the florida ecological network, has $128$ nodes and $2075$ edges. The best estimate is provided, chosen from among the most fragile value across all of the algorithms. In most cases, the METIS algorithm provides the best estimate for $\delta =0.5$, though the method with rewiring (which is both rewiring and iterative add back) provides the best estimate for the WS class of networks and the US Air 97 graph, as well as finds an identical value to METIS for the DM-LC network. 
    }
    \label{tab:RandomNets}
\end{table}
\begin{figure}[b]
\includegraphics[width=0.55\textwidth]{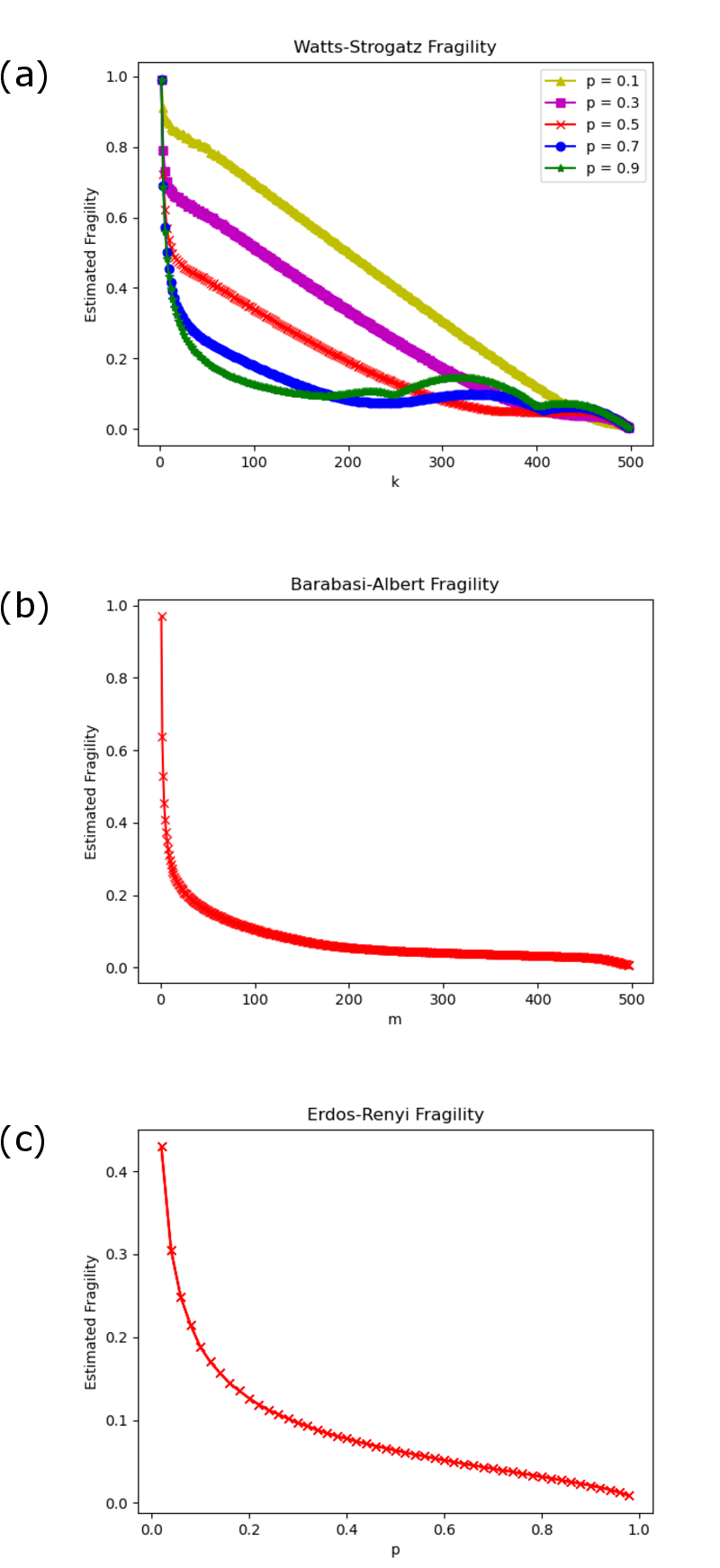}
\caption{\color{black} The estimated fragilities of 500 node networks with $\delta=0.5$, averaged over $100$ network realizations for each parameter value across (a) Watts-Strogatz networks with $p \in \{0.1,0.3,0.5,0.7,0.9\}$ and $k$ ranging from $2$ to $498$ uniformly spaced, (b) Barabasi-Albert networks with $m$ ranging from $1$ to $500$ uniformly spaced, and (c) Erd\H{o}s-R\'{e}nyi networks with $p$ ranging from $0.02$ to $0.98$.}
\label{fig:FragsParams}
\end{figure}
\paragraph*{}
In the supplementary videos, it can be seen that the various attack strategies lead to quite different degree distributions as measured by the Hellinger divergence. To further illustrate this point, in Table \ref{tab:RandomNets}, it is shown that the estimated average value of $\mathcal{F}_{0.5}$ is lower for both the minimum degree and edge betweenness attack strategies than the estimated value our method mixing both along with rewiring. {\color{black} For most of the scenarios, the best performing algorithm for estimation is METIS, with the exception of the Watts-Strogatz networks. METIS also offers a significant computational speed advantage over the other techniques, though neither it nor the Kernighan-Lin method are easily generalizable to all $\delta$ values.} This suggests that graphs are generally more fragile to edge removal than was previously understood. {\color{black} Interestingly, BS, ER, and WS all produced similar estimated fragilities, though the most fragile appeared to be the WS, indicating that such networks are generally more susceptible to bi-partioning via targeted attack.} Additionally, we note the high fragility of the Safegraph mall network, which suggests that removing a small number of edges in a person-to-person interaction network may quickly disintegrate the network. This may have implications for strategies for limiting epidemic spread among other applications.

\section{Conclusions}
In this work, we have presented a new measure for the fragility of a network to edge attacks. This measure differs from previous measures in that it does not rely on sequential attacks and does not represent a network as a single value. From this, a measure of robustness is derived. The concept of asymptotic robustness is presented. It is shown that in this new measure, the complete graph is robust. Additionally, a class of graphs which is sparser than the complete graph, {\color{black}the CEB}, is shown to be asymptotically robust. Finally, an algorithm for estimating the fragility of a general graph is presented. It is shown that graphs tend to be more fragile than previous methods would indicate and thus care should be taken when designing networks which may be subject to edge removal. 
\paragraph*{}
This work focused on the case in which we have global information about the edges of a graph. Frequently, only local information about the graph structure may be obtained. This suggests that in future work it may be beneficial to estimate the fragility when such global information is unavailable.

\begin{acknowledgments}
JF, EB and MB were all supported by DARPA, and EB is supported by the ARO, the NIH and ONR.
\end{acknowledgments}





\bibliography{apssamp}

@article{van2010,
  title={Exploring the brain network: a review on resting-state fMRI functional connectivity},
  author={Van Den Heuvel, Martijn P and Pol, Hilleke E Hulshoff},
  journal={European neuropsychopharmacology},
  volume={20},
  number={8},
  pages={519--534},
  year={2010},
  publisher={Elsevier}
}

@article{bassett2011,
  title={Dynamic reconfiguration of human brain networks during learning},
  author={Bassett, Danielle S and Wymbs, Nicholas F and Porter, Mason A and Mucha, Peter J and Carlson, Jean M and Grafton, Scott T},
  journal={Proceedings of the National Academy of Sciences},
  volume={108},
  number={18},
  pages={7641--7646},
  year={2011},
  publisher={National Academy of Sciences}
}

@article{park2013,
  title={Structural and functional brain networks: from connections to cognition},
  author={Park, Hae-Jeong and Friston, Karl},
  journal={Science},
  volume={342},
  number={6158},
  pages={1238411},
  year={2013},
  publisher={American Association for the Advancement of Science}
}

@article{fish2021,
  title={Entropic regression with neurologically motivated applications},
  author={Fish, Jeremie and DeWitt, Alexander and AlMomani, Abd AlRahman R and Laurienti, Paul J and Bollt, Erik},
  journal={Chaos: An Interdisciplinary Journal of Nonlinear Science},
  volume={31},
  number={11},
  year={2021},
  publisher={AIP Publishing}
}

@article{braess2005,
  title={On a paradox of traffic planning},
  author={Braess, Dietrich and Nagurney, Anna and Wakolbinger, Tina},
  journal={Transportation science},
  volume={39},
  number={4},
  pages={446--450},
  year={2005},
  publisher={INFORMS}
}

@article{motter2013,
  title={Spontaneous synchrony in power-grid networks},
  author={Motter, Adilson E and Myers, Seth A and Anghel, Marian and Nishikawa, Takashi},
  journal={Nature Physics},
  volume={9},
  number={3},
  pages={191--197},
  year={2013},
  publisher={Nature Publishing Group UK London}
}

@article{torres2017,
  title={Exploring topological effects on water distribution system performance using graph theory and statistical models},
  author={Torres, Jacob M and Duenas-Osorio, Leonardo and Li, Qilin and Yazdani, Alireza},
  journal={Journal of Water Resources Planning and Management},
  volume={143},
  number={1},
  pages={04016068},
  year={2017},
  publisher={American Society of Civil Engineers}
}

@article{guimera2004,
  title={Modeling the world-wide airport network},
  author={Guimera, Roger and Amaral, Lu{\i}s A Nunes},
  journal={The European Physical Journal B},
  volume={38},
  pages={381--385},
  year={2004},
  publisher={Springer}
}

@article{miller2015,
  title={Talking politics on Facebook: Network centrality and political discussion practices in social media},
  author={Miller, Patrick R and Bobkowski, Piotr S and Maliniak, Daniel and Rapoport, Ronald B},
  journal={Political Research Quarterly},
  volume={68},
  number={2},
  pages={377--391},
  year={2015},
  publisher={SAGE Publications Sage CA: Los Angeles, CA}
}

@article{himelboim2017,
  title={Classifying Twitter topic-networks using social network analysis},
  author={Himelboim, Itai and Smith, Marc A and Rainie, Lee and Shneiderman, Ben and Espina, Camila},
  journal={Social media+ society},
  volume={3},
  number={1},
  pages={2056305117691545},
  year={2017},
  publisher={SAGE Publications Sage UK: London, England}
}

@article{mackenzie2016,
  title={Adirondack landslides: History, exposures, and climbing},
  author={MacKenzie, Kevin B},
  journal={Adirondack Journal of Environmental Studies},
  volume={21},
  number={1},
  pages={13},
  year={2016}
}

@inproceedings{UCTE2006,
  title={Why has it happened again? Comparison between the UCTE blackout in 2006 and the blackouts of 2003},
  author={Bialek, Janusz W},
  booktitle={2007 IEEE Lausanne Power Tech},
  pages={51--56},
  year={2007},
  organization={IEEE}
}

@book{chung1997,
  title={Spectral graph theory},
  author={Chung, Fan RK},
  volume={92},
  year={1997},
  publisher={American Mathematical Soc.}
}

@article{chung2005,
  title={Laplacians and the Cheeger inequality for directed graphs},
  author={Chung, Fan},
  journal={Annals of Combinatorics},
  volume={9},
  pages={1--19},
  year={2005},
  publisher={Springer}
}

@article{cohen2000,
  title={Resilience of the internet to random breakdowns},
  author={Cohen, Reuven and Erez, Keren and Ben-Avraham, Daniel and Havlin, Shlomo},
  journal={Physical review letters},
  volume={85},
  number={21},
  pages={4626},
  year={2000},
  publisher={APS}
}

@article{cohen2001,
  title={Resilience of the internet to random breakdowns},
  author={Cohen, Reuven and Erez, Keren and Ben-Avraham, Daniel and Havlin, Shlomo},
  journal={Physical review letters},
  volume={85},
  number={21},
  pages={4626},
  year={2000},
  publisher={APS}
}

@article{newman2003,
  title={Mixing patterns in networks},
  author={Newman, Mark EJ},
  journal={Physical review E},
  volume={67},
  number={2},
  pages={026126},
  year={2003},
  publisher={APS}
}

@article{paul2004,
  title={Optimization of robustness of complex networks},
  author={Paul, Gerry and Tanizawa, T and Havlin, Shlomo and Stanley, H Eugene},
  journal={The European Physical Journal B},
  volume={38},
  number={2},
  pages={187--191},
  year={2004},
  publisher={Springer}
}

@article{schneider2011,
  title={Mitigation of malicious attacks on networks},
  author={Schneider, Christian M and Moreira, Andr{\'e} A and Andrade Jr, Jos{\'e} S and Havlin, Shlomo and Herrmann, Hans J},
  journal={Proceedings of the National Academy of Sciences},
  volume={108},
  number={10},
  pages={3838--3841},
  year={2011},
  publisher={National Academy of Sciences}
}

@book{louzada2015,
  title={Generating robust and efficient networks under targeted attacks},
  author={Louzada, Vitor HP and Daolio, Fabio and Herrmann, Hans J and Tomassini, Marco},
  year={2015},
  publisher={Springer}
}

@article{zeng2012,
  title={Enhancing network robustness against malicious attacks},
  author={Zeng, An and Liu, Weiping},
  journal={Physical Review E—Statistical, Nonlinear, and Soft Matter Physics},
  volume={85},
  number={6},
  pages={066130},
  year={2012},
  publisher={APS}
}

@article{duan2016,
  title={A comparative analysis of network robustness against different link attacks},
  author={Duan, Boping and Liu, Jing and Zhou, Mingxing and Ma, Liangliang},
  journal={Physica A: Statistical Mechanics and its Applications},
  volume={448},
  pages={144--153},
  year={2016},
  publisher={Elsevier}
}

@article{karger1996,
  title={A new approach to the minimum cut problem},
  author={Karger, David R and Stein, Clifford},
  journal={Journal of the ACM (JACM)},
  volume={43},
  number={4},
  pages={601--640},
  year={1996},
  publisher={ACM New York, NY, USA}
}

@article{newman2004finding,
  title={Finding and evaluating community structure in networks},
  author={Newman, Mark EJ and Girvan, Michelle},
  journal={Physical review E},
  volume={69},
  number={2},
  pages={026113},
  year={2004},
  publisher={APS}
}

@article{kernighan1970,
  title={An efficient heuristic procedure for partitioning graphs},
  author={Kernighan, Brian W and Lin, Shen},
  journal={The Bell system technical journal},
  volume={49},
  number={2},
  pages={291--307},
  year={1970},
  publisher={Nokia Bell Labs}
}

@article{karypis1997,
  title={A software package for partitioning unstructured graphs, partitioning meshes, and computing fill-reducing orderings of sparse matrices},
  author={Karypis, George and Kumar, Vipin},
  journal={University of Minnesota, Department of Computer Science and Engineering, Army HPC Research Center, Minneapolis, MN},
  volume={38},
  pages={7--1},
  year={1998}
}

@article{blondel2008,
  title={Fast unfolding of communities in large networks},
  author={Blondel, Vincent D and Guillaume, Jean-Loup and Lambiotte, Renaud and Lefebvre, Etienne},
  journal={Journal of statistical mechanics: theory and experiment},
  volume={2008},
  number={10},
  pages={P10008},
  year={2008},
  publisher={IOP Publishing}
}

@article{ma2016,
  title={Evolution of network robustness under continuous topological changes},
  author={Ma, Liangliang and Liu, Jing and Duan, Boping},
  journal={Physica A: Statistical Mechanics and its Applications},
  volume={451},
  pages={623--631},
  year={2016},
  publisher={Elsevier}
}

@article{liu2017,
  title={A comparative study of network robustness measures},
  author={Liu, Jing and Zhou, Mingxing and Wang, Shuai and Liu, Penghui},
  journal={Frontiers of Computer Science},
  volume={11},
  pages={568--584},
  year={2017},
  publisher={Springer}
}

@article{ren2018,
  title={Underestimated cost of targeted attacks on complex networks},
  author={Ren, Xiao-Long and Gleinig, Niels and Toli{\'c}, Dijana and Antulov-Fantulin, Nino},
  journal={Complexity},
  volume={2018},
  number={1},
  pages={9826243},
  year={2018},
  publisher={Wiley Online Library}
}

@article{heller2017,
  title={Unexpected failure of a Greedy choice Algorithm Proposed by Hoffman},
  author={Heller, Lauren and Sack, Andrew},
  journal={Int. J. Math. Comput. Sci},
  volume={12},
  number={2},
  pages={117--126},
  year={2017}
}

@article{jiang2014,
  title={A distributed routing for wireless sensor networks with mobile sink based on the greedy embedding},
  author={Jiang, Yisong and Shi, Weiren and Wang, Xiaogang and Li, Hongbing},
  journal={Ad Hoc Networks},
  volume={20},
  pages={150--162},
  year={2014},
  publisher={Elsevier}
}

@article{zahedinejad2014,
  title={Evolutionary algorithms for hard quantum control},
  author={Zahedinejad, Ehsan and Schirmer, Sophie and Sanders, Barry C},
  journal={Physical Review A},
  volume={90},
  number={3},
  pages={032310},
  year={2014},
  publisher={APS}
}

@online{safegraph,
  author = {Safegraph},
  title = {{Safegraph} Safegraph},
  year = 2021,
  url = {https://www.safegraph.com/covid-19-data-consortium.},
  urldate = {2021-04-27}
}

@inproceedings{Achalla20,
  title={Statistical methods for fast los detection for ranging and localization},
  author={Achalla, Monalisa and Mack, Kevin and Banavar, Mahesh K and Vanitha, M and Krishnamoorthi, Harish},
  booktitle={2020 International Conference on Emerging Trends in Information Technology and Engineering (ic-ETITE)},
  pages={1--5},
  year={2020},
  organization={IEEE}
}

@book{pardo2018,
  title={Statistical inference based on divergence measures},
  author={Pardo, Leandro},
  year={2018},
  publisher={Chapman and Hall/CRC}
}

@online{fish2022,
  author = {Jeremie Fish},
  title = {{Code} Code},
  year = 2022,
  url = {https://github.com/jefish003/NetworkFragility},
  urldate = {2022-08-29}
}

@article{shi2024local,
  title={Local dominance unveils clusters in networks},
  author={Shi, Dingyi and Shang, Fan and Chen, Bingsheng and Expert, Paul and L{\"u}, Linyuan and Stanley, H Eugene and Lambiotte, Renaud and Evans, Tim S and Li, Ruiqi},
  journal={Communications physics},
  volume={7},
  number={1},
  pages={170},
  year={2024},
  publisher={Nature Publishing Group UK London}
}

@article{peng2025unveiling,
  title={Unveiling explosive vulnerability of networks through edge collective behavior},
  author={Peng, Peng and Fan, Tianlong and Ren, Xiao-Long and L{\"u}, Linyuan},
  journal={Reliability Engineering \& System Safety},
  pages={111741},
  year={2025},
  publisher={Elsevier}
}

@article{morone2015influence,
  title={Influence maximization in complex networks through optimal percolation},
  author={Morone, Flaviano and Makse, Hern{\'a}n A},
  journal={Nature},
  volume={524},
  number={7563},
  pages={65--68},
  year={2015},
  publisher={Nature Publishing Group UK London}
}

@article{watts1998collective,
  title={Collective dynamics of ‘small-world’networks},
  author={Watts, Duncan J and Strogatz, Steven H},
  journal={nature},
  volume={393},
  number={6684},
  pages={440--442},
  year={1998},
  publisher={Nature Publishing Group}
}

@article{erdds1959random,
  title={On random graphs I},
  author={ERDdS, P and R\&wi, A},
  journal={Publ. math. debrecen},
  volume={6},
  number={290-297},
  pages={18},
  year={1959}
}

@article{barabasi1999emergence,
  title={Emergence of scaling in random networks},
  author={Barab{\'a}si, Albert-L{\'a}szl{\'o} and Albert, R{\'e}ka},
  journal={science},
  volume={286},
  number={5439},
  pages={509--512},
  year={1999},
  publisher={American Association for the Advancement of Science}
}

@inproceedings{nr,
     title={The Network Data Repository with Interactive Graph Analytics and Visualization},
     author={Ryan A. Rossi and Nesreen K. Ahmed},
     booktitle={AAAI},
     url={https://networkrepository.com},
     year={2015}
}

@article{hellinger1909neue,
  title={Neue begr{\"u}ndung der theorie quadratischer formen von unendlichvielen ver{\"a}nderlichen.},
  author={Hellinger, Ernst},
  journal={Journal f{\"u}r die reine und angewandte Mathematik},
  volume={1909},
  number={136},
  pages={210--271},
  year={1909},
  publisher={De Gruyter}
}

@article{torres2025heterochrony,
  title={Heterochrony in orthodenticle expression is associated with ommatidial size variation between Drosophila species},
  author={Torres-Oliva, Montserrat and Buchberger, Elisa and Buffry, Alexandra D and Kittelmann, Maike and Guerrero, Genoveva and Sumner-Rooney, Lauren and Gaspar, Pedro and Bullinger, Georg C and Jimenez, Javier Figueras and Casares, Fernando and others},
  journal={BMC biology},
  volume={23},
  number={1},
  pages={34},
  year={2025},
  publisher={Springer}
}

\end{document}